\documentclass{article}
\title{New descriptions of the weighted Reed-Muller codes and the homogeneous Reed-Muller codes}
\author{Harinaivo ANDRIATAHINY\\Mention: Mathematics and Computer Science,\\Domain: Sciences and Technologies,\\University of Antananarivo, Madagascar\\e-mail: hariandriatahiny@gmail.com\\\\Vololona Harinoro RAKOTOMALALA\\Mention: Meteorology,\\Domain: Sciences of the Engineer,\\Higher Polytechnic School of Antananarivo (ESPA),\\University of Antananarivo, Madagascar\\e-mail: volhrakoto@gmail.com}

\usepackage{amsmath,amssymb}
\usepackage{amscd,amsfonts}
\usepackage{amsthm}
\usepackage[english]{babel}
\usepackage[T1]{fontenc}
\usepackage[ansinew]{inputenc}
\swapnumbers
\theoremstyle{plain}
\newtheorem{thm}{Theorem}[section]
\newtheorem{prop}[thm]{Proposition}

\newtheorem{cor}[thm]{Corollary}

\theoremstyle{definition}

\newtheorem{rem}[thm]{Remark}

\DeclareMathOperator{\card}{card}

\begin{document}

\maketitle

\begin{abstract}
We give a description of the weighted Reed-Muller codes over a prime field in a modular algebra. A description of the homogeneous Reed-Muller codes in the same ambient space is presented for the binary case. A decoding procedure using the Landrock-Manz method is developed.
\end{abstract}
\section{Introduction}
It is well known that the Generalized Reed-Muller (GRM) codes of length $p^m$ over the prime field $\mathbb{F}_{p}$ can be viewed as the radical powers of the modular algebra $A=\mathbb{F}_{p}[X_{0},\ldots ,X_{m-1}]/(X_{0}^{p}-1,\ldots ,X_{m-1}^{p}-1)$ ([1],[4],[5]). $A$ is isomorphic to the group algebra $\mathbb{F}_{p}[\mathbb{F}_{p^m}]$.\\
The weighted Reed-Muller codes and the homogeneous Reed-Muller codes are classes of codes in the Reed-Muller family. The Jennings basis are used to describe the GRM codes over $\mathbb{F}_{p}$. We utilize the elements of the Jennings basis for the description of the weighted Reed-Muller codes and the homogeneous Reed-Muller codes in $A$. P. Landrock and O. Manz developed a decoding algorithm for the binary Reed-Muller codes in [9]. We use here the same method for the binary homogeneous Reed-Muller codes.\\
The weighted Reed-Muller codes can be considered as a generalization of the GRM codes. Some classes of the weighted Reed-Muller codes are algebraic-geometric codes.\\
The homogeneous Reed-Muller codes are subcodes of the GRM codes. In general, they have a much better minimum distance than the GRM codes.\\
We give, in section 2, the definition and some properties of the weighted Reed-Muller codes. We consider here the affine case. In section 3, a description of the weighted Reed-Muller codes over $\mathbb{F}_{p}$ in the quotient ring $A$ is presented. In section 4, we recall the definition and the parameters of the homogeneous Reed-Muller codes. In section 5, we describe the homogeneous Reed-Muller codes over the two elements field $\mathbb{F}_{2}$ in $A$ (with $p=2$). In section 6, we use the Landrock-Manz method to construct a decoding procedure for the homogeneous Reed-Muller codes in the binary case. In section 7, an example is given.

\section{Weighted Reed-Muller codes}
The definition and the properties of the weighted Reed-Muller codes presented in this section are from [11].
Let $\mathbb{F}_{q}$ the field of $q=p^r$ elements where $p$ is a prime number and $r\geq 1$ is an integer.
Let $(\mathbb{F}_{q})^{m}$ be the $m$-dimensional affine space defined over $\mathbb{F}_{q}$. $\mathbb{F}_{q}[Y_{0},Y_{1},\ldots ,Y_{m-1}]$ is the ring of polynomials in $m$ variables with coefficients in $\mathbb{F}_{q}$. If we attach to each variables $Y_{i}$ a natural number $w_{i}$, called weight of $Y_{i}$, we speak about the ring of weighted polynomials, $W\mathbb{F}_{q}[Y_{0},Y_{1},\ldots ,Y_{m-1}]$. The weighted degree of $F\in W\mathbb{F}_{q}[Y_{0},Y_{1},\ldots ,Y_{m-1}]$, is defined as
\begin{equation*}
\deg_{\varpi}(F)=\deg_{\varpi}(F(Y_{0},\ldots ,Y_{m-1}))=\deg (F(Y_{0}^{w_{0}},\ldots ,Y_{m-1}^{w_{m-1}})),
\end{equation*}
where $\deg$ is the usual degree.\\
We will, without loss of generality, always assume that the weights are ordered $w_{0}\leq w_{1}\leq \ldots \leq w_{m-1}$.
Consider the evaluation map
\begin{equation}
\begin{aligned}
\phi :\quad  W\mathbb{F}_{q}[Y_{0},\ldots ,Y_{m-1}]&\longrightarrow (\mathbb{F}_{q})^{q^{m}} \\
        F &\longmapsto \phi(F)=(F(P_{1}),\ldots ,F(P_{n}))
\end{aligned}
\end{equation}
where $P_{1},\ldots ,P_{n}$ ($n=q^{m}$) is an arbitrary ordering of the elements of $(\mathbb{F}_{q})^{m}$.\\
For $w_{0}= w_{1}= \ldots = w_{m-1}=1$ we have the following definition.\\
The Generalized Reed-Muller codes of order $\nu$ ($1\leq \nu \leq m(q-1)$) and length $n=q^{m}$ is defined by
\begin{equation*}
C_{\nu}(m,q)=\phi(V(\nu)),
\end{equation*}
where
\begin{equation*}
V(\nu)=\{F\in \mathbb{F}_{q}[Y_{0},\ldots ,Y_{m-1}]\mid \deg(F)\leq\nu\}.
\end{equation*}
Let $\omega$ be a natural number and $\{w_{0},\ldots ,w_{m-1}\}$ be weights corresponding to the ring of weighted polynomials $W\mathbb{F}_{q}[Y_{0},\ldots ,Y_{m-1}]$. The weighted Reed-Muller codes $WRMC_{\omega}(m,q)$ of weighted order $\omega$ and length $n=q^{m}$, corresponding to the weights $\{w_{0},\ldots ,w_{m-1}\}$ is defined by
\begin{equation}
WRMC_{\omega}(m,q)=\phi(V_{\varpi}(\omega)),
\end{equation}
where
\begin{equation}
V_{\varpi}(\omega)=\{F\in W\mathbb{F}_{q}[Y_{0},\ldots ,Y_{m-1}]\mid \deg_{\varpi}(F)\leq \omega\}.
\end{equation}
For a polynomial $F\in \mathbb{F}_{q}[Y_{0},\ldots ,Y_{m-1}]$, $\overline{F}$ denotes the reduced form of $F$, i.e. the polynomial of lowest degree equivalent to $F$ modulo the ideal $(Y_{i}^{q}-Y_{i},i=0,\ldots m-1)$. For any subset $M$ of $\mathbb{F}_{q}[Y_{0},\ldots ,Y_{m-1}]$, the set $\overline{M}$ denotes the set of reduced elements of $M$.
\begin{rem}
For $F\in \mathbb{F}_{q}[Y_{0},\ldots ,Y_{m-1}]$, we have
\begin{enumerate}
\item for every $P \in (\mathbb{F}_{q})^{m}: F(P)= \overline{F}(P)$.
\item if $F(P)=0$ for all $P\in (\mathbb{F}_{q})^{m}$, then $\overline{F}=0$.
\end{enumerate}
\end{rem}
Given natural numbers $\omega,\nu$, and a set of weights $\{w_{0},\ldots ,w_{m-1}\}$ such that $1\leq \nu\leq m(q-1)$ and $1\leq \omega\leq (q-1)\sum_{i=1}^{m}w_{i}$.\\
Let
\begin{equation*}
\nu_{max}(\omega)=Q'(q-1)+R'
\end{equation*}
where
\begin{equation*}
Q'=max\{Q\mid \omega\geq \sum_{i=0}^{Q}(q-1)w_{i}\}
\end{equation*}
and
\begin{equation*}
R'=max\{R\mid \omega\geq \sum_{i=0}^{Q'}(q-1)w_{i}+Rw_{Q'+1}\}.
\end{equation*}
\begin{thm}
Given a natural number $\omega$ and a set of ordered weights\\$\{w_{0},\ldots ,w_{m-1}\}$ such that $\omega \leq (q-1)\sum_{i=0}^{m-1}w_{i}$. The code $WRMC_{\omega}(m,q)$ is an $\mathbb{F}_{q}$-linear $[q^{m},k,d]$ code with
\begin{equation*}
k=\card(\{(e_{0},\ldots ,e_{m-1})\mid \sum_{i=0}^{m-1}w_{i}e_{i}\leq \omega, 0\leq e_{i}<q\})
\end{equation*}
and
\begin{equation*}
d=q^{m-Q-1}(q-R)
\end{equation*}
where $Q$ and $R$ are given by
\begin{equation*}
\nu_{max}(\omega)=Q(q-1)+R,
\end{equation*}
\end{thm}
with $0\leq R< q-1$.
\begin{rem}
The set of monomials
\begin{equation*}
\{\prod_{i=0}^{m-1}Y_{i}^{e_{i}}\mid \sum_{i=0}^{m-1}w_{i}e_{i}\leq \omega,0\leq e_{i}<q\}
\end{equation*}
is a basis of $\overline{V_{\varpi}(\omega)}$.
\end{rem}

\section{Description of the weighted Reed-Muller codes in $A$}
Consider the modular algebra
\begin{equation*}
A=\mathbb{F}_{p}[X_{0},X_{1},\ldots ,X_{m-1}]/(X_{0}^{p}-1,\ldots ,X_{m-1}^{p}-1)
\end{equation*}
and the ideal
\begin{equation*}
I=\left(X_{0}^{p}-1,\ldots,X_{m-1}^{p}-1\right)
\end{equation*}
of the polynomial ring $\mathbb{F}_{p}[X_{0},\ldots,X_{m-1}]$, where $\mathbb{F}_{p}$ is the prime field of $p$ (a prime number) elements.\\
Set $x_0=X_0+I,\ldots,x_{m-1}=X_{m-1}+I$.
Let us fix an order on the set of monomials
\begin{equation*}
\left\{x_{0}^{i_0}\ldots x_{m-1}^{i_{m-1}} \mid 0\leq i_{0},\ldots,i_{m-1}\leq p-1\right\}.
\end{equation*}
Then
\begin{equation}
A=\left\{\sum_{i_{0}=0}^{p-1}\cdots\sum_{i_{m-1}=0}^{p-1}a_{i_{0}\ldots i_{m-1}}x_{0}^{i_0}\ldots x_{m-1}^{i_{m-1}} \mid a_{i_0\ldots i_{m-1}}\in \mathbb{F}_{p}\right\}.
\end{equation}
And we have the following identification:\\
$A\ni\sum_{i_0=0}^{p-1}\cdots\sum_{i_{m-1}=0}^{p-1}a_{i_0\ldots i_{m-1}}x_0^{i_0}\ldots x_{m-1}^{i_{m-1}} \longleftrightarrow (a_{i_0\ldots i_{m-1}})_{0\leq i_{0},\ldots,i_{m-1}\leq p-1}\in \left(\mathbb{F}_{p}\right)^{p^m}$.\\
Hence the modular algebra $A$ is identified with $\left(\mathbb{F}_{p}\right)^{p^m}$.\\
$P(m,p)$ denotes the vector space of the reduced polynomials in $m$ variables over $\mathbb{F}_{p}$:
\begin{equation*}\label{redpoly}
\left\{P(Y_{0},\ldots,Y_{m-1})=\sum_{i_{0}=0}^{p-1}\cdots\sum_{i_{m-1}=0}^{p-1}u_{i_0\ldots i_{m-1}}Y_0^{i_0}\ldots Y_{m-1}^{i_{m-1}} \mid u_{i_0\ldots i_{m-1}}\in \mathbb{F}_{p}\right\}.
\end{equation*}
Consider  a set of weights $\{w_{0},\ldots ,w_{m-1}\}$ and let $\omega$ be an integer such that $0\leq\omega\leq (p-1)(w_{0}+\ldots +w_{m-1})$.\\
When considering $P(m,p)$ and $A$ as vector spaces over $\mathbb{F}_{p}$, we have the following isomorphism:
\begin{equation}\label{isom}
\begin{aligned}
\psi :\quad  P(m,p)&\longrightarrow A \\
        P(Y_{0},\ldots,Y_{m-1}) &\longmapsto \sum_{i_0=0}^{p-1}\cdots\sum_{i_{m-1}=0}^{p-1}P(i_0,\ldots ,i_{m-1})x_0^{i_0}\ldots x_{m-1}^{i_{m-1}}
\end{aligned}
\end{equation}
The set
\begin{equation}
B:=\left\{(x_{0}-1)^{i_{0}}\ldots (x_{m-1}-1)^{i_{m-1}} \mid 0\leq i_{0},\ldots,i_{m-1}\leq p-1\right\}
\end{equation}
is called the Jennings basis of $A$.\\
Set $\left[0,p^m-1\right]=\left\{0,1,2,\ldots ,p^m-1\right\}$.\\
Let $i\in \left[0,p^m-1\right]$. Consider its $p$-adic expansion
\begin{equation*}
i=\sum_{k=0}^{m-1}i_{k}p^{k}
\end{equation*}
with $0\leq i_{k}\leq p-1$ for all $k=0,\ldots ,m-1$.\\
We need the following notations and definitions:\\
$\underline{i}:=(i_{0},\ldots,i_{m-1})$,\\
the $p$-weight of $i$ is defined by $wt_{p}(i):=\sum_{k=0}^{m-1}i_{k}$,\\
and the $p$-weight of $i$ with respect to the set of weights $\{w_{0},\ldots ,w_{m-1}\}$ is defined by
\begin{equation}
Wwt_{p}(i):=\sum_{k=0}^{m-1}i_{k}w_{k}.
\end{equation}
$\underline{j}\leq\underline{i}$ if $j_{l}\leq i_{l}$ for all $l=0,1,\ldots ,m-1$ where $\underline{j}:=(j_{0},\ldots,j_{m-1})\in ([0,p-1])^{m}$,\\
$\underline{x}:=(x_{0},\ldots,x_{m-1})$,\\
$\underline{x}^{\underline{i}}:=x_0^{i_0}\ldots x_{m-1}^{i_{m-1}}$.\\
\noindent Consider the polynomial
\begin{equation}\label{jennpoly}
B_{\underline{i}}(\underline{x}):=(x_{0}-1)^{i_0}\ldots (x_{m-1}-1)^{i_{m-1}} \in A.
\end{equation}
The following proposition is from [1].
\begin{prop}\label{jenpol}
We have $H_{\underline{i}}(\underline{Y})=\psi^{-1}(B_{\underline{i}}(\underline{x}))$, where $\psi$ is the isomorphism defined in (\ref{isom}), i.e.
\begin{equation*}
B_{\underline{i}}(\underline{x})=\sum_{\underline{j}\leq \underline{i}}H_{\underline{i}}(\underline{j})\underline{x}^{\underline{j}}
\end{equation*}
where
\begin{equation*}\label{prodinterpf}
H_{\underline{i}}(\underline{Y}):=\prod_{l=0}^{m-1}H_{i_l}(Y_l)
\end{equation*}
and
\begin{equation*}
H_i(Y) = \alpha_i \prod_{j=1}^{p-1-i} (Y+j),
\end{equation*}
with $\alpha_i = - i! \mbox{ mod } p$.
\end{prop}
\begin{cor}\label{weightdeg}
We have
\begin{equation*}
\deg_{\varpi}(H_{\underline{i}}(\underline{Y}))=(p-1)\sum_{l=0}^{m-1}w_{l}-Wwt_{p}(i).
\end{equation*}
\end{cor}
\noindent We now present a description of the weighted Reed-Muller code $WRMC_{\omega}(m,p)$ in the algebra $A$.
\begin{thm}
Consider  a set of weights $\{w_{0},\ldots ,w_{m-1}\}$ and let $\omega$ be an integer such that $0\leq \omega\leq (p-1)\sum_{l=0}^{m-1}w_{l}$. Then, the set
\begin{equation*}
B_{\omega}:=\{(x_{0}-1)^{i_{0}}\ldots (x_{m-1}-1)^{i_{m-1}}\mid 0\leq i_{k}\leq p-1, \sum_{k=0}^{m-1}w_{k}i_{k}\geq (p-1)\sum_{k=0}^{m-1}w_{k}-\omega\}
\end{equation*}
forms a linear basis of the weighted Reed-Muller code $WRMC_{\omega}(m,p)$ over $\mathbb{F}_{p}$ in $A$.
\end{thm}
\begin{proof}
It is clear that $B_{\omega}$ is a set of linearly independant elements because $B_{\omega}\subseteq B$.\\
Let $B_{\underline{i}}(\underline{x}):=(x_{0}-1)^{i_0}\ldots (x_{m-1}-1)^{i_{m-1}} \in B_{\omega}$, i.e. $0\leq i_{k}\leq p-1$, for all $k=0,\ldots ,m-1$, and $\sum_{k=0}^{m-1}w_{k}i_{k}\geq (p-1)\sum_{k=0}^{m-1}w_{k}-\omega$.\\
By the Proposition \ref{jenpol} and the Corollary \ref{weightdeg}, we have $B_{\underline{i}}(\underline{x})=\sum_{\underline{j}\leq \underline{i}}H_{\underline{i}}(\underline{j})\underline{x}^{\underline{j}}$ with $H_{\underline{i}}(\underline{Y})=\prod_{l=0}^{m-1}H_{i_l}(Y_l)$, $H_i(Y) = \alpha_i \prod_{j=1}^{p-1-i} (Y+j)$, and $\alpha_i = - i! \mbox{ mod } p$.\\
We have $\deg_{\varpi}(H_{\underline{i}}(\underline{Y}))=(p-1)\sum_{l=0}^{m-1}w_{l}-Wwt_{p}(i)\leq \omega$.\\
Thus $H_{\underline{i}}(\underline{Y})\in V_{\varpi}(\omega)$.\\
Therefore, $B_{\underline{i}}(\underline{x})\in WRMC_{\omega}(m,p)$.\\
It is clear that $\dim_{\mathbb{F}_{p}}(WRMC_{\omega}(m,p))=\card(\{i\in [0,p^m-1]\mid Wwt_{p}(i)\leq \omega\})$.\\
On the other hand, we have $\card(B_{\omega})=\card(\{i\in[0,p^m-1]\mid Wwt_{p}(i)\geq (p-1)\sum_{k=0}^{m-1}w_{k}-\omega\})$.\\
Consider the bijection
\begin{equation*}
\begin{aligned}
\theta : [0,p^m-1] &\longrightarrow [0,p^m-1]\\
         i=\sum_{k=0}^{m-1}i_{k}p^{k}&\longmapsto \theta(i)=\sum_{k=0}^{m-1}(p-1-i_{k})p^{k}.
\end{aligned}
\end{equation*}
We have $Wwt_{p}(\theta(i))=\sum_{k=0}^{m-1}w_{k}(p-1-i_{k})=(p-1)\sum_{k=0}^{m-1}w_{k}-Wwt_{p}(i)$, i.e. $Wwt_{p}(i)=(p-1)\sum_{k=0}^{m-1}w_{k}-Wwt_{p}(\theta(i))$.\\
Thus, we have $Wwt_{p}(i)\leq \omega \Longleftrightarrow Wwt_{p}(\theta(i))\geq (p-1)\sum_{k=0}^{m-1}w_{k}-\omega$.\\
Hence, $\card(\{i\in[0,p^m-1]\mid Wwt_{p}(i)\leq \omega\})=\card(\{i\in[0,p^m-1]\mid Wwt_{p}(i)\geq (p-1)\sum_{k=0}^{m-1}w_{k}-\omega\})$.
\end{proof}
The following Corollary is the famous result of Berman-Charpin ([1],[4],[5]).
\begin{cor}
Consider the weights $w_{0}=\ldots =w_{m-1}=1$ and an integer $\omega$ such that $0\leq \omega\leq m(p-1)$. Then, the set
\begin{equation*}
B_{\omega}:=\{(x_{0}-1)^{i_{0}}\ldots (x_{m-1}-1)^{i_{m-1}}\mid 0\leq i_{k}\leq p-1, \sum_{k=0}^{m-1}i_{k}\geq m(p-1)-\omega\}
\end{equation*}
forms a linear basis of the GRM code $C_{\omega}(m,p)=P^{m(p-1)-\omega}$ over $\mathbb{F}_{p}$, where $P$ is the radical power of $A$.
\end{cor}

\section{The homogeneous Reed-Muller codes}
In this section, we recall the definition and some properties of the homogeneous Reed-Muller codes [3],[10]. $\mathbb{F}_{q}$ denote the field of $q=p^r$ elements with $p$ a prime number and $r\geq 1$ an integer. For $n=q^m-1$, let $\{0,P_{1},\ldots ,P_{n}\}$ be the set of points in $(\mathbb{F}_{q})^m$ ordered in a fixed order.\\
Let $\mathbb{F}_{q}[Y_{0},\ldots,Y_{m-1}]_{d}^{0}$ be the vector space of homogeneous polynomials in $m$ variables over $\mathbb{F}_{q}$ of degree $d$.\\
Now $d$ denote an integer such that $0\leq d\leq m(q-1)$. The $d$th order homogeneous Reed-Muller (HRM) codes of length $q^m$ over $\mathbb{F}_{q}$ is defined as
\begin{equation}
HRMC_{d}(m,q):=\{(F(0),F(P_{1}),\ldots,F(P_{n}))\mid F\in \mathbb{F}_{q}[Y_{0},\ldots,Y_{m-1}]_{d}^{0}\}.
\end{equation}
Thus $HRMC_{d}(m,q)$ is a proper subcode of the GRM code $C_{d}(m,q)$.\\
The following theorem can be found in [3].
\begin{thm}
Let $d$ such that $1\leq d\leq (m-1)(q-1)$. The HRM code $HRMC_{d}(m,q)$ is an $[n+1,k,\delta]$ linear code with $n+1=q^m$,
\begin{equation*}
k=\sum_{t\equiv d mod(q-1),0<t\leq d}\sum_{j=0}^{m}(-1)^j \binom{m}{j}\binom{t-jq+m-1}{t-jq},
\end{equation*}
and
\begin{equation*}
\delta=(q-1)(q-s)q^{m-r-2},
\end{equation*}
where $d-1=r(q-1)+s$ and $0\leq s<q-1$.
\end{thm}

\section{Description of the binary HRM codes in $A$}
First, we recall some results in the Proposition \ref{jenpol} for the special case $p=2$. In this section, we consider the ambiant space
\begin{equation*}
A=\mathbb{F}_{2}[X_{0},\ldots ,X_{m-1}]/(X_{0}^{2}-1,\ldots ,X_{m-1}^{2}-1).
\end{equation*}
We have
\begin{equation*}
B_{\underline{i}}(\underline{x})=(x_{0}-1)^{i_0}\ldots (x_{m-1}-1)^{i_{m-1}}=\sum_{\underline{j}\leq \underline{i}}H_{\underline{i}}(\underline{j})\underline{x}^{\underline{j}}
\end{equation*}
where $0\leq i_{k}\leq 1$, for all $k$,
\begin{equation*}
H_{\underline{i}}(\underline{Y}):=\prod_{l=0}^{m-1}H_{i_l}(Y_l)
\end{equation*}
and
\begin{equation*}
H_i(Y) = \alpha_i \prod_{j=1}^{1-i} (Y+j),
\end{equation*}
with $\alpha_i = - i! \mbox{ mod } 2$.\\
Note that $B_{(1,1,\ldots,1)}(\underline{x})=(x_{0}-1)^1\ldots (x_{m-1}-1)^1=\hat{1}$ is the "all one" word.\\
Let $d$ be an integer such that $0\leq d\leq m$. The $d$th order homogeneous Reed-Muller (HRM) codes of length $2^m$ over $\mathbb{F}_{2}$ is defined as
\begin{equation*}
HRMC_{d}(m,2):=\{(F(0),F(P_{1}),\ldots,F(P_{n}))\mid F\in \mathbb{F}_{2}[Y_{0},\ldots,Y_{m-1}]_{d}^{0}\},
\end{equation*}
where $n=2^m-1$.
We now give the description of the binary HRM code $HRMC_{d}(m,2)$ in $A$.
\begin{thm}
Let $d$ be an integer such that $1\leq d\leq m$. The set
\begin{equation*}
\{(x_{0}-1)^{i_0}\ldots (x_{m-1}-1)^{i_{m-1}}+\hat{1}\mid 0\leq i_{k}\leq 1,\quad m>\sum_{k=0}^{m-1}i_{k}\geq m-d\}
\end{equation*}
forms a linear basis for the binary HRM code $HRMC_{d}(m,2)$.
\end{thm}
\begin{proof}
Let $d$ such that $1\leq d\leq m$.\\
Consider the element
\begin{equation*}
B_{\underline{i}}(\underline{x})+\hat{1}=(x_{0}-1)^{i_0}\ldots (x_{m-1}-1)^{i_{m-1}}+\hat{1},
\end{equation*}
where $0\leq i_{k}\leq 1$ for all $k$ and $m>\sum_{k=0}^{m-1}i_{k}\geq m-d$.\\
Set $D(\underline{i}):=\{\underline{j}\in(\{0,1\})^{m}\mid \underline{j}\leq \underline{i}\}$ and $C(\underline{i}):=(\{0,1\})^{m}- D(\underline{i})$.\\
We have $H_{\underline{i}}(\underline{j})=0$ for $\underline{j}\in C(\underline{i})$.\\
Thus 
\begin{equation*}
B_{\underline{i}}(\underline{x})=\sum_{\underline{j}\leq \underline{i}}H_{\underline{i}}(\underline{j})\underline{x}^{\underline{j}}=\sum_{\underline{j}\in (\{0,1\})^m}H_{\underline{i}}(\underline{j})\underline{x}^{\underline{j}}\quad.
\end{equation*}
We have
\begin{equation*}
H_{\underline{i}}(\underline{Y}):=\prod_{l=0}^{m-1}H_{i_l}(Y_l),
\end{equation*}
where
\begin{equation}\label{polybase}
H_1(Y) = 1, H_0(Y) = Y+1.
\end{equation}
Since $\sum_{k=0}^{m-1}i_{k}\geq m-d$, then $B_{\underline{i}}(\underline{x})\in P^{m-d}$ where $P$ is the radical of the modular algebra $A$.\\
And since $P^{m-d}=C_{d}(m,2)$, then $H_{\underline{i}}(\underline{Y})\in P_{d}(m,2)$ where $P_{d}(m,2)$ is a linear space generated by the set
\begin{equation}\label{multipoly}
\{Y_0^{i_0}\ldots Y_{m-1}^{i_{m-1}}\mid 0\leq i_{k}\leq 1,0\leq \sum_{k=0}^{m-1}i_{k}\leq d\}.
\end{equation}
We have
\begin{equation*}
B_{\underline{i}}(\underline{x})+\hat{1}=\sum_{\underline{j}\in (\{0,1\})^m}(H_{\underline{i}}(\underline{j})+1)\underline{x}^{\underline{j}}\quad.
\end{equation*}
By (\ref{polybase}) and (\ref{multipoly}), we have $H_{\underline{i}}(\underline{Y})+1\in \mathbb{F}_{2}[Y_{0},\ldots,Y_{m-1}]_{d}^{0}$.\\
Note that $\mathbb{F}_{2}[Y_{0},\ldots,Y_{m-1}]_{d}^{0}$ is a linear space generated by the set
\begin{equation*}
S:=\{Y_0^{i_0}\ldots Y_{m-1}^{i_{m-1}}\mid 0\leq i_{k}\leq 1,0< \sum_{k=0}^{m-1}i_{k}\leq d\}.
\end{equation*}
Thus $B_{\underline{i}}(\underline{x})+\hat{1}\in HRMC_{d}(m,2)$.\\
Note also that $\sum_{k=0}^{m-1}i_{k}=m$ if and only if $i_{k}=1$ for all $k=0,\ldots,m-1$.\\
Set $R:=\{(x_{0}-1)^{i_0}\ldots (x_{m-1}-1)^{i_{m-1}}+\hat{1}\mid 0\leq i_{k}\leq 1, m>\sum_{k=0}^{m-1}i_{k}\geq m-d\}$.\\
We will show that $\dim_{\mathbb{F}_{2}}(HRMC_{d}(m,2))=\card(R)$.\\
We have $\dim_{\mathbb{F}_{2}}(HRMC_{d}(m,2))=\dim_{\mathbb{F}_{2}}(\mathbb{F}_{2}[Y_{0},\ldots,Y_{m-1}]_{d}^{0})=\card(S)$.\\
Consider the bijection
\begin{equation*}
\begin{aligned}
\beta :\quad\quad (\{0,1\})^{m} &\longrightarrow (\{0,1\})^{m}\\
         (i_{0},\ldots,i_{m-1})&\longmapsto (1-i_{0},\ldots,1-i_{m-1})
\end{aligned}
\end{equation*}
Set $R':=\{\underline{i}=(i_{0},\ldots,i_{m-1})\in (\{0,1\})^{m}\mid \sum_{k=0}^{m-1}i_{k}\geq m-d\}$\\
and $S':=\{\underline{i}=(i_{0},\ldots,i_{m-1})\in (\{0,1\})^{m}\mid \sum_{k=0}^{m-1}i_{k}\leq d\}$.\\
It is clear that $S'=\beta(R')$. Thus $\card(R')=\card(S')$.\\
Since $\card(R)=\card(R')-1$ and $\card(S)=\card(S')-1$, then $\card(R)=\card(S)$.
\end{proof}

\section{Decoding procedure for the binary HRM codes}
In this section, we will follow Landrock-Manz as in [9].\\
Let $d$ be an integer such that $1\leq d\leq m$. The HRM code $HRMC_{d}(m,2)$ is of type $\left[2^m,\sum_{t=1}^{d}\binom{m}{t},2^{m-d}\right]$ over $\mathbb{F}_{2}$.\\
Set $b(\{i_{1},\ldots,i_{t}\}):=(x_{i_{1}}-1)\ldots(x_{i_{t}}-1)$,\\
where $\{i_{1},\ldots,i_{t}\}\subseteq\{0,1,\ldots,m-1\}$.\\
$B_{m-d}:=\{b(\eta)+\hat{1}\mid \eta\subseteq \{0,1,\ldots,m-1\},m > \card(\eta)\geq m-d\}$ is a linear basis of $HRMC_{d}(m,2)$.\\
General results of the following Proposition can be found in [2].
\begin{prop}\label{Jennprop}
We have
\begin{enumerate}
\item $b(\{\})=1$.
\item \begin{equation*}
b(\eta).b(\kappa)=\begin{cases}
         0& \text{if $\eta\cap\kappa\neq\{\}$},\\
         b(\eta\cup\kappa)& \text{otherwise}.
\end{cases}
\end{equation*}
\item The weight of the codeword $w(b(\{i_{1},\ldots,i_{t}\}))=2^t$.
\item $b(\{0,1,\ldots,m-1\})=\hat{1}$ the "all one" word.
\item $\hat{1}.b(\{\eta\})=0$ if $\eta\neq \{\}$.
\end{enumerate}
\end{prop}
\noindent Set $\eta^c:=\{0,1,\ldots,m-1\}-\eta$.\\
Let $c\in HRMC_{d}(m,2)$ be a transmitted codeword and $v\in A$ the received vector, where
\begin{equation*}
A=\mathbb{F}_{2}[X_{0},\ldots ,X_{m-1}]/(X_{0}^{2}-1,\ldots ,X_{m-1}^{2}-1)
\end{equation*}
Since $HRMC_{d}(m,2)$ is $(2^{m-d-1}-1)$-error correcting, we write $v=c+f$ with $w(f)\leq 2^{m-d-1}-1$.\\
We have
\begin{equation*}
c=\sum_{m-d\leq\card(\eta)<m,\; \eta\subseteq \{0,1,\ldots,m-1\}}\tau(\eta)(b(\eta)+\hat{1})
\end{equation*}
with $\tau(\eta)\in \mathbb{F}_{2}$.\\
We now present the decoding procedure to determine the coefficients $\tau(\eta)$.
\underline{Step1:}\\
Let $\kappa$ be a subset of $\{0,1,\ldots,m-1\}$ such that $\card(\kappa)=m-d$. We have
\begin{equation*}
\begin{aligned}
v.b(\kappa^c)&=(c+f).b(\kappa^c)\\
&=(\sum_{m-d\leq\card(\eta)<m,\; \eta\subseteq \{0,1,\ldots,m-1\}}\tau(\eta)(b(\eta)+\hat{1})+f).b(\kappa^c)\\
&=\sum_{m-d\leq\card(\eta)<m,\; \eta\subseteq \{0,1,\ldots,m-1\}}\tau(\eta)(b(\eta).b(\kappa^c)+\hat{1}.b(\kappa^c))+f.b(\kappa^c)\\
&=\tau(\kappa).\hat{1}+f.b(\kappa^c)
\end{aligned}
\end{equation*}
We have
$w(f.b(\kappa^c))\leq w(f).w(b(\kappa^c))\leq (2^{m-d-1}-1).2^d=2^{m-1}-2^d<2^{m-1}=\frac{1}{2}2^m$.\\
Then it is easy to see that
\begin{center}
$\tau(\kappa)=0$ if and only if $w(v.b(\kappa^c))<2^{m-1}$.
\end{center}
We next subtract $\tau(\kappa).(b(\kappa)+\hat{1})$ from $v$ and obtain $v'=v+\tau(\kappa).(b(\kappa)+\hat{1})=c'+f$ where $c'=c+\tau(\kappa).(b(\kappa)+\hat{1}) \in HRMC_{d}(m,2)$.\\
Consider another set $\eta$ of cardinality $\card(\eta)=m-d$, multiply $v'$ by $b(\eta^c)$ and find so $\tau(\eta)$. Having eventually run through all sets of cardinality $m-d$, we end up with $v''=c''+f$, where $c''\in HRMC_{d-1}(m,2)$.\\
\underline{Step2:}\\
We now fix a set $\kappa$ of cardinality $m-d+1$ and by using the same technique as in the first step, we can find the coefficient $\tau(\kappa)$. We repeat the same treatment for all set $\eta$ of cardinality $m-d+1$. We eventually determine
\begin{equation*}
c=\sum_{m-d\leq\card(\eta)<m,\; \eta\subseteq \{0,1,\ldots,m-1\}}\tau(\eta)(b(\eta)+\hat{1}).
\end{equation*}
If another step is needed, we must pick a set $\kappa$ of cardinality $m-d+2$ and determine $\tau(\kappa)$, and we continue in this way.

\section{An example}
Consider the binary HRM code $HRMC_{1}(5,2)$. This code is of type $\left[32,5,16\right]$.\\
Set $\widetilde{E_{l}}=\prod_{k\neq l}(x_{k}-1)=b(\{l\}^c)$, where $l\in \{0,1,2,3,4\}$.\\
The set $\{\widetilde{E_{l}}+\hat{1}\mid l=0,1,2,3,4\}$ forms a linear basis for $HRMC_{1}(5,2)$.\\
Let $c=\sum_{i=0}^{4}\tau_{i}.(\widetilde{E_{i}}+\hat{1})$ be a transmitted codeword with $\tau_{i}\in\mathbb{F}_{2}$ and $v\in A=\mathbb{F}_{2}[X_{0},X_{1},\ldots ,X_{4}]/(X_{0}^{2}-1,\ldots ,X_{4}^{2}-1)$ the received vector. Since $HRMC_{1}(5,2)$ is $7$-error correcting, we write $v=c+f$ with $w(f)\leq 7$.\\
We have\\
$v.(x_j-1)=(c+f).(x_j-1)=c.(x_j-1)+f.(x_j-1)=\tau_j.\hat{1}+f.(x_j-1)$ for $j=0,1,2,3,4$.\\
Since $w(f.(x_j-1))\leq w(f).w((x_j-1))\leq 7*2=14<16$, then we have the following Proposition
\begin{prop}\label{coeftest}
\begin{equation*}
\tau_j=0 \quad if \;and\; only\; if \quad w(v.(x_j-1))<16
\end{equation*}
\end{prop}
By multiplying $v$ with $(x_j-1)$ for $j=0,1,2,3,4$ and utilizing the Proposition \ref{coeftest}, we obtain the coefficients $\tau_0,\tau_1,\tau_2,\tau_3,\tau_4$ of $c$.

\end{document}